\documentclass[14pt]{article}

\usepackage{arxiv}

\usepackage[utf8]{inputenc} 
\usepackage[T1]{fontenc}    
\usepackage{hyperref}       
\usepackage{url}            
\usepackage{booktabs}       
\usepackage{amsfonts}       
\usepackage{nicefrac}       
\usepackage{microtype}      
\usepackage{lipsum}
\usepackage{graphicx}
\graphicspath{ {./images/} }

\usepackage{amsthm}
\usepackage{amsmath}
\usepackage{amssymb}
\usepackage{mathtools}
\usepackage{xcolor}
\usepackage{dsfont}
\usepackage{verbatim}

\newcommand{\N}{\mathbb{N}}
\newcommand{\R}{\mathbb{R}}

\newcommand{\Prob}{\mathbb{P}}
\newcommand{\Han}{\mathbb{H}}
\DeclareMathOperator{\ind}{\mathds{1}}
\DeclareMathOperator{\E}{\mathbb{E}}
\DeclarePairedDelimiter{\abs}{\lvert}{\rvert}
\DeclarePairedDelimiter{\norm}{\lvert\lvert}{\rvert\rvert_{\infty}}

\newtheorem{theorem}{Theorem}
\newtheorem{corollary}[theorem]{Corollary}

\newtheorem{counterexample}{Counterexample}
\newtheorem{definition}[theorem]{Definition}

\newtheorem{remark}{Remark}
\newtheorem{example}{Example}

\numberwithin{theorem}{section}
\numberwithin{remark}{section}
\numberwithin{example}{section}
\numberwithin{counterexample}{section}

\DeclareMathOperator*{\argmax}{arg\,max}
\newcommand{\setplayers}{N}
\newcommand{\player}{i} 
\newcommand{\setactions}{A}
\newcommand{\paction}{a} 
\newcommand{\ppaction}{a}
\newcommand{\maction}{x}

\newcommand{\caction}{\gamma}
\newcommand{\utility}{u}
\newcommand{\mixedset}{\Delta}
\newcommand{\game}{G}
\newcommand{\Vp}{V_{pure}}
\newcommand{\vp}{v_{pure}}
\newcommand{\Vm}{V_{mixed}}
\newcommand{\vm}{v_{mixed}}
\newcommand{\vc}{v_{corr}}

\newcommand{\vh}{v_{h}}
\newcommand{\rounds}{M}
\newcommand{\emp}{Z}
\newcommand{\Corr}{CED}

\title{Playing Against No-Regret Players}

\author{
 Maurizio D'Andrea \\
  ANITI and Toulouse School of Economics, UT Capitole\\
 \texttt{maurizio.dandrea@tse-fr.eu}
}

\begin{document}
\maketitle

\begin{abstract}

In increasingly different contexts, it happens that a human player has to interact with artificial players who make decisions following decision-making algorithms. How should the human player play against these algorithms to maximize his utility? Does anything change if he faces one or more artificial players? The main goal of the paper is to answer these two questions. 

Consider n-player games in normal form repeated over time, where we call the human player optimizer, and the $(n-1)$ artificial players, learners. We assume that learners play no-regret algorithms, a class of algorithms widely used in online learning and decision-making. In these games, we consider the concept of Stackelberg equilibrium.
In a recent paper, Deng, Schneider, and Sivan have shown that in a 2-player game the optimizer can always guarantee an expected cumulative utility of at least the Stackelberg value per round. In our first result, we show, with counterexamples, that this result is no longer true if the optimizer has to face more than one player. Therefore, we generalize the definition of Stackelberg equilibrium introducing the concept of correlated Stackelberg equilibrium. Finally, in the main result, we prove that the optimizer can guarantee at least the correlated Stackelberg value per round. Moreover, using a version of the strong law of large numbers, we show that our result is also true  almost surely for the optimizer utility instead of the optimizer's expected utility.
\end{abstract}


\section{Introduction}
Consider a $n$-player finite game repeated over time. In practice in online learning, strategies are often computed through decision-making algorithms and a class of these algorithms widely used and deeply studied are the No-Regret algorithms. A No-Regret strategy ensures that looking at the history of the strategies played, there is no action that if it had been played at each round would have guaranteed a better payoff.
How should the human player (Optimizer) behave to face No-Regret learners? Is there any gameplay for the optimizer to increase his utility?  Does anything change if he faces one or more artificial players? Answering these questions is the main focus of the paper.

An important role in our work is played by the concept of Stackelberg equilibrium of the game. The Stackelberg version of our game is a game where the optimizer chooses his strategy first and communicates his choice; after observing the strategy of the optimizer, the learners play a Nash equilibrium. Moreover, we define the Stackelberg value as the worst/best utility for the optimizer in any Stackelberg equilibrium. In a recent paper \cite{A1} Deng, Schneider and Sivan have proved that in games with an optimizer and a learner, if the learner plays No-Regret the optimizer can always guarantee an expected utility of at least the Stackelberg value per round. Starting from this result, the first question we ask ourselves is: is the result valid even in a game with more than one learner? The answer to this question is No. Indeed, we show that it is possible to find a simple counterexample with an optimizer and two learners where the learners play No-Regret but the optimizer receives an expected utility strictly lower than the Stackelberg value per round.

Given this negative result, we extend the definition of Stackelberg equilibrium so that it fits better games with more than one No-Regret learner. We first recall well-known results for No-Regret strategies. In a game in normal form where every player plays No-Regret, the empirical average distribution converges almost surely to the set of correlated equilibrium distributions. Hence, we define the correlated Stackelberg equilibrium as an equilibrium where the optimizer plays first and communicates his strategy, then the learners play a correlated equilibrium instead of a Nash equilibrium. We define the correlated Stackelberg value as the worst/best utility for the optimizer in any correlated Stackelberg equilibrium. Thus, in our main result, we prove that in a $n$-player game with an optimizer and $n-1$ learners, if the learners play No-Regret the optimizer can guarantee an expected utility of at least the correlated Stackelberg value per round. Finally using a version of the strong law of large numbers we show that our main result is also true almost surely for the utility of the optimizer and not only for the expected utility of the optimizer.

\paragraph{Other related work}
The notion of No-Regret learning is one of the most famous and deeply studied in online learning and decision making. The first notion of Regret, called External-Regret, was given by Hannan in \cite{Hannan} and it was subsequently taken up in numerous other papers (see \cite{NR1,NR2,NR3,NR4}). Later another stronger Regret notion called Internal-Regret was introduced earlier in \cite{NIR}. Several No-Internal Regret algorithms have been created, for example in \cite{NIR2,HartMasColell}. This notation is linked with the concept of correlated equilibrium introduced for the first time by Aumann in \cite{AUMANN197467}. The paper \cite{NIR1} of Hart and Mas-Colell prove that in a n-player game in normal form where every player plays No-Internal Regret the empirical average distribution converges almost surely to the set of correlated equilibrium distributions. Moreover, if every player plays No-External Regret the empirical average distribution converges almost surely to the Hanna number of papers that study the behavior of No-regret algorithm in long term games, see for example \cite{flokas2020noregret,10.1257/000282803322655581}.

\section{Model}
\label{sec:Model}
\subsection{Games and Equilibria}
In this paper we consider n-player games, where each player has a finite set of strategies. Let $\setplayers=\{1,...,n\}$ be the set of players and let $\setactions_{\player}$ be the finite set of actions of player $\player$. Denote with $\setactions=\prod_{\player=1}^n\setactions_{\player}$ the set of strategy profiles. For every player $\player$ we define  $\utility_{\player}:\setactions \to \R$ the utility function of player $\player$, i.e. if the players play the strategy profile $\paction\in\setactions$ the payoff of player $\player$ is $\utility_{\player}(\paction)$. Players can randomize their strategies, so we define with $\mixedset_{\player}$ the set of mixed strategies for the player $\player$ and denote $\mixedset=\prod_{\player\in\setplayers}\mixedset_{\player}$. As usual we may linearly extend the utility functions to the set of mixed strategies.
For the rest of the article, we will refer to the n-th player as the optimizer and we will call the remaining players Learners and we call $\game$ the game described above. 

\begin{definition}
A strategy $\maction_{\player}\in\mixedset_{\player}$ is a {\sl Best-Reply} for player $\player$ against $\maction_{-\player}\in\mixedset_{-\player}$ \footnote{We use the notation $x_{-i}$ to say that we consider all elements except element $i$; i.e. $x_{-i}=(x_1,...,x_{i-1},x_{i+1},...,x_n)$. Moreover we denote $\Delta_{-i}=\prod_{j\neq i}\Delta_j$ } if $$\maction_{\player}\in\argmax_{z\in\mixedset_i}\utility_{\player}(z,\maction_{-\player}).$$
We denote  $BR^i(\maction_{-i})$ the set of all the best replies of player $\player$ against the mixed action $\maction_{-\player}\in\mixedset_{-\player}$.
\end{definition}

We are now ready to define the notion of equilibrium that we consider in the game $\game$ where there are $n-1$ learners and $1$ optimizer. This notion was introduced for the first time by Von Stackelberg in \cite{stackelberg}. A Stackelberg equilibrium is an equilibrium where the followers play a Nash equilibrium once they have observed the optimizer’s mixed strategy. We refer to an optimistic Stackelberg equilibrium when the followers play a Nash equilibrium that
maximizes the optimizer’s utility, and to a pessimistic Stackelberg equilibrium when they play a Nash equilibrium that minimizes the optimizer utility.

\begin{definition}
\label{SV}
The profile of strategies $(\paction,\maction_n)\in \setactions_{-n} \times \mixedset_n$ is a {\sl pure Optimistic Stackelberg Equilibrium} ($POSE$) if
$$\utility_n(\paction,\maction_n)=\max_{y\in \mixedset_n}\max_{z\in \setactions_{-n}}\{ \utility_n(z,y)|\ z_i\in BR^i(z_{-i},y)\forall i=1,...,n-1\}.$$
We call the value $\Vp:=\utility_n(\paction,\maction_n)$ the pure Optimistic Stackelberg value of the game if such an equilibrium exists. 

The profile of strategies $(\paction,\maction_n)\in \setactions_{-n} \times \mixedset_n$ is a {\sl pure Pessimistic Stackelberg Equilibrium} ($PPSE$) if
$$\utility_n(\paction,\maction_n)=\max_{y\in \mixedset_n}\min_{z\in \setactions_{-n}}\{ \utility_n(z,y)|\ z_i\in BR^i(z_{-i},y) \forall i=1,...,n-1\}.$$
We call the value $\vp:=\utility_n(\paction,\maction_n)$ the pure Pessimistic Stackelberg value of the game if such an equilibrium exists.
 \end{definition}

However, in a game $\game$ with $n>2$ pure Stackelberg Equilibria may not exist. Therefore we extend the definition of Stackelberg equilibrium in a natural way so that its existence is guaranteed in any n-player game.

\begin{definition}
The profile of strategies $\maction\in \mixedset$ is a {\sl mixed Optimistic Stackelberg Equilibrium} ($MOSE$) if
$$\utility_n(\maction)=\max_{y \in \mixedset_n}\max_{z\in \mixedset_{-n}}\{ \utility_n(z,y)|\ z_i\in BR^i(z_{-i},y) \forall i=1,...,n-1\}.$$
We call the value $\Vm:=\utility_n(\maction)$ the mixed Optimistic Stackelberg value of the game.

The profile of strategies $\maction\in \mixedset$ is a {\sl mixed Pessimistic Stackelberg Equilibrium} ($MPSE$) if
$$\utility_n(\maction)=\max_{y\in \mixedset_n}\min_{z\in \mixedset_{-n}}\{ \utility_n(z,y)|\ z_i\in BR^i(z_{-i},y) \forall i=1,...,n-1\}.$$
We call the value $\vm:=\utility_n(\maction)$ the mixed Pessimistic Stackelberg value of the game.
\end{definition}

\begin{remark}
The values defined above respect the following chain of inequalities
$$\vm \leq \vp \leq \Vp \leq \Vm.$$
Moreover for $n=2$ we have $\Vp=\Vm$ and $\vp=\vm$.
\end{remark}

\subsection{No-Regret}

From here on we are interested in the game $\game$ repeated over time. Hence, we denote by $\paction^t_{\player}$  the realization at time $t$ of the player $\player$ . Let $\paction^t=(\paction^t_{1},...,\paction^t_n)$ and we assume that players' utilities are additive over time.
A strategy for the player $\player$ is an element $\sigma_{\player}=(\sigma_{\player}^t)_{t\geq1}$ where, for each $t\in\N$, $\ \sigma_{\player}^t$ is a map from $(A)^{t-1}$ to $\mixedset_{\player}$. Denote by $\sigma=(\sigma_1,....,\sigma_n)\in\Sigma$ the profile of strategies. For each stage $t$, $\utility_{\player}(\paction_1^t,...,\paction_{n}^t,)$ is the random realized payoff at stage t for the player $\player$. Finally, we denote with $\Prob_{\sigma}$ the probability induced by the infinite sequences $(a^t)_{t\geq0}$.

To describe a possible optimal behavior of the players in these types of repeated games the no-regret algorithms are the most commonly used.

\begin{definition}
A strategy $\sigma_i$ of the learner $\player$ has {\sl No-Regret} if for every profile of strategies $\sigma_{-\player}$,
$$\left(\max_{\paction\in\setactions_{\player}}\frac{1}{\rounds}\sum_{t=1}^{\rounds}(\utility_{\player}(\paction,\paction^t_{-\player})-\utility_{\player}(\paction^t))\right)_{+} \xrightarrow{M\to\infty} 0 \qquad \Prob_{\sigma} \ a.s. \footnote{We use the notation $(f)_{+}=\max\{f,0\}$ }$$
\end{definition}

\begin{definition}
A strategy $\sigma_i$ of the learner $\player$ has {\sl No-Internal Regret} if for every profile of strategies $\sigma_{-\player}$,
$$\max_{\paction, b\in\setactions_{\player}}\frac{1}{\rounds}\sum_{\substack{t\in\{1,...,\rounds\}, \\ a_{\player}^t=\paction}} (\utility_{\player}(b,\paction^t_{-\player})-\utility_{\player}(\paction,\paction^t_{-\player}))\xrightarrow{M\to\infty}0 \qquad \Prob_{\sigma} \ a.s.$$
\end{definition}

Intuitively, No-Regret guarantees that strategies perform well against the best possible action, while No-Internal
Regret guarantees that strategies perform as well as the best possible action over each subset of rounds where the same action is played. Thus, it is clear that the set of No-Internal Regret strategies is a subset of the set of No-Regret strategies.

We conclude this section with two well known results for No-Regret and No-Internal Regret strategies. For every $t$, let $\emp_t \in \Delta(\setactions)$ be the {\sl empirical average distribution} of the $n$-tuples of strategies played until time $t$, that is for every $\paction\in \setactions$
\begin{equation}
\emp_t(\paction):=\frac{1}{t} \abs{\{s\leq t : \ \paction^{s}=\paction\}}.
\end{equation}

\begin{theorem}[ Hart and Mas-Colell\cite{NIR1}] \label{J}
Let $\game$ be a game and let $\Corr(\game)$ be the set of correlated equilibrium distribution of the game $\game$; that is
\begin{equation}
\Corr(\game):=\left\{\Phi\in\Delta(\setactions): \sum_{\paction_{-\player}\in \setactions_{-\player}} \Phi(\paction)\utility_{\player}(\paction'_{\player},\paction_{-\player})\leq \sum_{\paction_{-\player}\in \setactions_{-\player}} \Phi(\paction)\utility_{\player}(\paction_{\player},\paction_{-\player}), \quad \forall \player\in\setplayers,\ \paction_{\player},\  \paction'_{\player}\in\setactions_{\player}\right\}.
\end{equation}
Then, if each player of the game $\game$ follows some No-Internal Regret procedure, the distance from the empirical distribution of moves to $\Corr(\game)$ converges a.s. to $0$.
\end{theorem}

\begin{theorem}
Let $\game$ be a game and let $\Han(\game)$ be the Hannan set of the game $\game$; that is
\begin{equation}
\Han(\game):=\left\{\Phi\in\Delta(\setactions): \sum_{\paction\in \setactions} \Phi(\paction)\utility_{\player}(\paction'_{\player},\paction_{-\player})\leq \sum_{\paction\in \setactions} \Phi(\paction)\utility_{\player}(\paction), \quad \forall \player\in\setplayers, \paction'_{\player}\in\setactions_{\player}\right\}.
\end{equation} 
Then, if each player of the game $\game$ follows
some No-Regret procedure, the distance from the empirical distribution of moves to the Hannan set $\Han(\game)$ converge a.s. to $0$. 
\end{theorem}

\section{Preliminary Results}

We start presenting a result given by Deng, Schneider, and Sivan in \cite{A1} which shows how in a two-player game, the optimizer can achieve an average utility per round arbitrarily close to the pure optimistic Stackelberg value against a no-regret learner.

\begin{definition}
A strategy $\paction_{\player}\in\setactions_{\player}$ is {\sl very weakly dominated} if  there exists $\maction_{\player}\in\mixedset(\setactions_{\player}-\{\paction_{\player}\})$ such that for all $\paction_{-\player}\in\setactions_{-\player}$, $\utility_{\player}(\maction_{\player},\ppaction_{-\player})\geq\utility_{\player}(\ppaction_{\player},\ppaction_{-\player})$.
\end{definition}

\begin{theorem}[ Theorem 4 \cite{A1}] \label{T1}
Let $\game$ be a game with two players ($n=2$) repeated $\rounds$ times. Let $\Vp$ be the pure optimistic Stackelberg value of the game $\game$. Assume that the learner doesn't have very weakly dominated strategies. If the learner is playing no-regret, then for every $\epsilon>0$ the optimizer can guarantee at least $(\Vp-\epsilon)\rounds-o(\rounds)$ expected cumulative utility.
\end{theorem}

The goal of this section is to show that this result is false in a game with more than two players.
To prove it we show how to construct a class of no-internal regret strategies in the following specific case.

\begin{example}
\label{ex1}
Let $\game$ be a $2$-player game $\game$ with set of actions $\setactions_1=\{T,B\}$, $\setactions_2=\{L,R\}$ and utility function for player $1$ given by the matrix
\[
\bordermatrix{
 & L & R \cr
T & 1 & 0 \cr
B & 0 & 1 
}.
\]
Then a strategy $\sigma_1$ of player 1  with no-internal regret can be constructed as follows. Play arbitrarily at stage 1. At the end of stage $t$, for $t\geq1$, compute the vector of regret $\overline{r}^{\ t}=\begin{pmatrix} 0 & \overline{r}^{\ t}_{T,B} \\ \overline{r}^{\ t}_{B,T} & 0 \end{pmatrix} $ with
$$ \overline{r}^{\ t}_{T,B}=\frac{1}{t} \sum_{\substack{m\in\{1,...,t\},\\ \paction^m_1=T}} \left(\mathds{1}_{a_2^m=R}-\mathds{1}_{a_2^m=L}\right)$$
and
$$ \overline{r}^{\ t}_{B,T}=\frac{1}{t} \sum_{\substack{m\in\{1,...,t\},\\ \paction^m_1=B}} \left(\mathds{1}_{a_2^m=L}-\mathds{1}_{a_2^m=R}\right).$$
And at stage $t+1$, play $T$ with probability $p\in[0,1]$ such that
\begin{equation}\label{prob}
p\cdot( \overline{r}^{\ t}_{T,B})_{+}=(1-p)\cdot( \overline{r}^{\ t}_{B,T})_{+}.
\end{equation}
\end{example}

Now we are ready to prove that Theorem \ref{T1} is false for $n>2$.  

\begin{counterexample} \label{CE1}
Consider a $3$-player game where $A_1=\{T, B\}$, $A_{2}=\{L, R\}$, $A_3=\{E\}$ and the payoff matrix is 
\[
\bordermatrix{
 & L & R \cr
T & (1,1,0) & (0, 0, 0) \cr
B & (0,0,-1) & (1,1,1)  
}.
\]
We can define a no-regret strategy $\sigma_1$ for Player 1 as in the Example \ref{ex1} with the assumption that at the first step he plays $T$ and if $p$ in \eqref{prob} is arbitrary he chooses $p=1$. 
Moreover, we can repeat the argument for player 2 reversing the role of $\ T$ with $L$ and $B$ with $R$. Thus we have $\sigma_1$ and $\sigma_2$ no-regret strategies respectively for player 1 and player 2.
Then
$$r_{T,B}^1=-1 \quad and \quad r_{B,T}^1=0 \Rightarrow \sigma_{1}^2=T$$
$$r_{L,R}^1=-1 \quad and \quad r_{R,L}^1=0 \Rightarrow \sigma_{2}^2=L.$$
Repeating the reasoning we have  $\sigma^t_{1}=T$,  $\sigma^t_{2}=L$ for every $t\geq1$.  Hence
$$\E\left(\sum_{t=1}^{\rounds}\utility_{3}(a_1^t,a_2^t,E)\right)=\sum_{t=1}^{\rounds}\utility_{3}(T,L,E)=0.$$
However  in this game $\Vp=1$, so Theorem \ref{T1} is false for $n>2$.
\end{counterexample}

The natural question that now arises is: can Theorem \ref{T1} be true for one of the other Stackelberg values defined in \ref{SV}? To answer this question, we will modify the previous counterexample so that the optimizer can not guarantee the mixed Pessimistic Stackelberg value per round.

\begin{counterexample} \label{CE2}

Consider a 3-player game $\game$ as in the Counterexample \ref{CE1} but with payoffs given by 
\[
\bordermatrix{
 & L & R \cr
T & (1,1,0) & (0, 0, 1) \cr
B & (0,0,-1) & (1,1,0)  
}.
\]
In this game it is easy to see that $\vm=0$. Assume that player 1 and player 2 play no-regret as in the Counterexample \ref{CE1} with the difference that if at the stage $t$ the choice of $p$ in   \eqref{prob} is arbitrary, player 1 play 
$$p=\begin{cases} 1 & if \ \sigma^{t-1}=(B,L) \\ 0 & if \  \sigma^{t-1}=(T,L) \ or \  \sigma^{t-1}=(B,R) \\ arbitrary & otherwise \end{cases}$$ 
and Player 2 plays
$$p=\begin{cases} 1 & if \ \sigma^{t-1}=(B,R) \ or \ \sigma^{t-1}=(B,L)  \\ 0 & if \   \sigma^{t-1}=(T,L) \\ arbitrary & otherwise. \end{cases}$$ 
Then
$$r_{T,B}^1=-1 \quad and \quad r_{B,T}^1=0 \Rightarrow \sigma_{1}^2=B$$
$$r_{L,R}^1=-1 \quad and \quad r_{R,L}^1=0 \Rightarrow \sigma_{2}^2=R.$$
$$r_{T,B}^2=-\frac{1}{2} \quad and \quad r_{B,T}^2=-\frac{1}{2} \Rightarrow \sigma_{1}^3=B$$
$$r_{L,R}^2=-\frac{1}{2} \quad and \quad r_{R,L}^2=-\frac{1}{2} \Rightarrow \sigma_{2}^3=L.$$
$$r_{T,B}^3=-\frac{1}{3} \quad and \quad r_{B,T}^3=0 \Rightarrow \sigma_{1}^4=T$$
$$r_{L,R}^3=0 \quad and \quad r_{R,L}^3=-\frac{1}{3} \Rightarrow \sigma_{2}^4=L,$$
and the process repeats. Hence, this couple of no-regret strategies induce the following sequence of action: $(T,L), \ (B,R), \ (B,L), \ (T,L), \ (B,R), \ (B,L),...$
and the cumulative payoff for the optimizer is

$$\frac{1}{\rounds}\E\left(\sum_{t=1}^{\rounds}\utility_{3}(a_1^t,a_2^t,E)\right)=-\frac{1}{3}<\vm.$$
\end{counterexample}

\section{Main Result}
In the counterexamples \ref{CE1} and \ref{CE2} the learners have played a 2 player game and their No-Internal Regret behavior has generated a sequence of actions whose frequency represents a correlated equilibrium of the game. Therefore, in this section we generalize the definition of Stackelberg equilibrium using the concept of correlated equilibrium instead of Nash equilibrium and we show that the optimizer can guarantee the value of this {\sl correlated Stackelberg equilibrium}.

\begin{definition}
The profile of strategies $(\caction,\alpha)\in \Delta(\setactions_{-n}) \times \mixedset_n$ is a {\sl correlated Pessimistic Stackelberg Equilibrium} ($CPSE$) if
$$\utility_n(\caction,\alpha)=\max_{y\in \mixedset_n}\min_{\Phi \in CED(\game_{y})}\utility_n(\Phi,y),$$
Where $\game_{y}$ is the $n-1$ player game in which the payoffs are given by $\utility_{\player}(\cdot,y)$; in this game players do not observe realizations of actions of player n. We call the value $\vc:=\utility_n(\caction,\alpha)$ the correlated Pessimistic Stackelberg value of the game.
\end{definition}
Similarly, we can define the Hannan Stackelberg equilibrium as follow.
\begin{definition}
The profile of strategies $(\caction,\alpha)\in \Delta(\setactions_{-n}) \times \mixedset_n$ is an {\sl Hannan Pessimistic Stackelberg Equilibrium} ($HPSE$) if
$$\utility_n(\caction,\alpha)=\max_{y\in \mixedset_n}\min_{\Phi \in \Han(\game_{y})}\utility_n(\Phi,y),$$
We call the value $\vh:=\utility_n(\caction,\alpha)$ the Hannan Pessimistic Stackelberg value of the game.
\end{definition}

\begin{theorem} \label{corrStack}
Let $\vc$ be the correlated pessimistic Stackelberg value of the game $\game$. If the learners are playing no-internal regret, then for every $\epsilon>0$ the optimizer can guarantee at least $(\vc-\epsilon)\rounds$ expected cumulative utility for $\rounds$ large enough.
\end{theorem}

To prove Theorem \ref{corrStack} we recall the version of the strong law of large numbers for uncorrelated random variables.

\begin{theorem}[{\bf Strong Law of Large Number}]
Let $\{X_t\}_{t\in \N}$ be a sequence of random variable such that $\E[X_t]=0$ and $Var[X_t]\leq C<\infty$. If for every $t,s$ in $\N$  with $t\neq s$, $\E[X_tX_s]=0$ then
$$\frac{\sum_{t=1}^{M}X_t}{M}\xrightarrow{M\to\infty}0 \quad a.s.$$
\end{theorem}

\begin{corollary} \label{CorSLLN}
Let $\game$ be a n-player game. Let $\sigma_n$ be the strategy of player $n$ that consists of playing $\alpha\in\Delta_n$ in every period. For each profile of strategies $\sigma_{-n}$ of the other players, and for every $i\in\setplayers$
$$ \frac{1}{M}\sum_{t=1}^M (\utility_i(a_{-n}^t,a_n^t)-\utility_i(a_{-n}^t,\alpha)) \xrightarrow{M\to\infty}0 \quad \Prob_{\sigma} \ a.s.$$
\end{corollary}

Finally, we are ready to prove our main result.

\begin{proof} [Proof Theorem \ref{corrStack}]
Let $(\caction,\alpha)\in\Delta(\setactions_{-n}) \times \mixedset_n$ be a correlated Stackelberg equilibrium. Fix $\epsilon>0$. Suppose that at each round the optimizer plays the strategy $\alpha$. Since the followers play no internal regret; for every $i\in\setplayers$ and every $a,b\in\setactions_i$
\begin{align} \label{NIR2}
\frac{1}{\rounds}\sum_{t=1}^{\rounds}\ind_{\paction_i^t=a}(\utility_{\player}(b,\paction_{-i}^t)-\utility_{\player}(a,\paction_{-i}^t))\leq o(1) \qquad \Prob_{\sigma} \ a.s.
\end{align}

Fix $a,b\in\setactions_i$, then \footnote{We use the notation $a_{-i,n}$ to say that we consider all elements except element $i$ and $n$; i.e. $a_{-i,n}=(a_1,...,a_{i-1},a_{i+1},...,a_{n-1})$.}

\begin{align}
\frac{1}{\rounds}\sum_{t=1}^{\rounds}\ind_{\paction_i^t=a}(\utility_{\player}(b,\paction_{-i,n}^t,\alpha)-\utility_{\player}(a,\paction_{-i,n}^t,\alpha))&= \frac{1}{\rounds}\sum_{t=1}^{\rounds}\ind_{\paction_i^t=a}(\utility_{\player}(b,\paction_{-i,n}^t,\alpha)-\utility_{\player}(b,\paction_{-i}^t)) \nonumber \\
&+\frac{1}{\rounds}\sum_{t=1}^{\rounds}\ind_{\paction_i^t=a}(\utility_{\player}(b,\paction_{-i}^t)-\utility_{\player}(a,\paction_{-i}^t)) \nonumber \\
&+\frac{1}{\rounds}\sum_{t=1}^{\rounds}\ind_{\paction_i^t=a}(\utility_{\player}(a,\paction_{-i}^t)-\utility_{\player}(a,\paction_{-i,n}^t,\alpha)) \nonumber \\
&= \frac{1}{\rounds}\sum_{t=1}^{\rounds}(v_{\player}(\paction_{-n}^t,\alpha)-v_{\player}(\paction^t)) \label{p1} \\
&+\frac{1}{\rounds}\sum_{t=1}^{\rounds}\ind_{\paction_i^t=a}(\utility_{\player}(b,\paction_{-i}^t)-\utility_{\player}(a,\paction_{-i}^t)) \label{p2} \\
&+\frac{1}{\rounds}\sum_{t=1}^{\rounds}(w_{\player}(\paction^t)-w_{\player}(\paction^t,\alpha)) \label{p3}
\end{align}

Where $v_{\player}(\paction^t):=\ind_{\paction_i^t=a}\utility_{\player}(b,\paction_{-i}^t)$ and $w_{\player}(\paction^t):=\ind_{\paction_i^t=a}\utility_{\player}(a,\paction_{-i}^t)$. Now, \eqref{p1} and \eqref{p3} go to $0$ $\Prob_{\sigma}$ a.s. by Corollary \ref{CorSLLN}, and using the inequality \eqref{NIR2} we obtain that for every $i\in\setplayers$ and for every $a,b\in\setactions_{\player}$
$$\frac{1}{\rounds}\sum_{t=1}^{\rounds}\ind_{\paction_i^t=a}(\utility_{\player}(b,\paction_{-i,n}^t,\alpha)-\utility_{\player}(a,\paction_{-i,n}^t,\alpha))\leq o(1) \qquad \Prob_{\sigma} \ a.s.$$

Hence, denoting with $Z_M \in \Delta(\setactions_{-n})$ the empirical distribution; for every $i\in\setplayers$ and for every $a,b\in\setactions_{\player}$
$$\sum_{\paction_{-i,n}\in\setactions_{-i,n}} Z_{\rounds}(a,\paction_{-i,n})(\utility_{\player}(b,\paction_{-i,n},\alpha)-\utility_{\player}(a,\paction_{-i,n},\alpha))\leq o(1) \qquad \Prob_{\sigma} \ a.s.$$
So, 
$$dist(Z_{\rounds}, CED(\game_{\alpha}))\xrightarrow{M\to\infty} 0 \qquad \Prob_{\sigma} \ a.s.$$

By dominated convergence theorem 
$$\E_{\sigma}[dist(Z_{\rounds}, CED(\game_{\alpha}))]\xrightarrow{M\to\infty} 0;$$
therefore for every $\delta>0$ there exist $\rounds_0$ such that

$$\E[dist(Z_{\rounds}, CED(\game_{\alpha}))]\leq\delta \qquad \forall \rounds\geq \rounds_0,$$

Denote by $\Psi_t$ the projection of $Z_t$ to the set $CED(\game_{\alpha})$, then
\begin{align}
\E\left[ \frac{1}{\rounds}\sum_{t=1}^{\rounds}\utility_n(\paction_{-n}^t, \paction_n^t)\right]&=\E\left[ \frac{1}{\rounds}\sum_{t=1}^{\rounds}\utility_n(\paction_{-n}^t, \alpha)\right]=\E\left[\sum_{\paction_{-n}\in\setactions_{-n}}Z_{\rounds}(\paction_{-n})\utility_n(\paction_{-n},\alpha)\right] 
 \nonumber \\
&\geq\E\left[\sum_{\paction_{-n}\in\setactions_{-n}}(Z_{\rounds}(\paction_{-n})-\Psi_{\rounds}(\paction_{-n})+\Psi_{\rounds}(\paction_{-n}))\utility_n(\paction_{-n},\alpha)\right] \nonumber \\
&\geq\E\left[\sum_{\paction_{-n}\in\setactions_{-n}}\Psi_{\rounds}(\paction_{-n})\utility_n(\paction_{-n},\alpha)\right]-\E\left[\sum_{\paction_{-n}\in\setactions_{-n}}(\Psi_{\rounds}(\paction_{-n})-Z_{\rounds}(\paction_{-n}))\utility_n(\paction_{-n},\alpha)\right] \nonumber \\
&\geq \E[\utility_n(\Psi_{\rounds},\alpha)] - \E\left[\sum_{\paction_{-n}\in\setactions_{-n}}\abs{\Psi_{\rounds}(\paction_{-n})-Z_{\rounds}(\paction_{-n})}\cdot\abs{\utility_n(\paction_{-n},\alpha)}\right]  \nonumber \\
&\geq \E[\utility_n(\Psi_{\rounds},\alpha)] - \norm{\utility_{n}}\E\left[\sum_{\paction_{-n}\in\setactions_{-n}}\abs{\Psi_{\rounds}(\paction_{-n})-Z_{\rounds}(\paction_{-n})}\right]  \nonumber \\
&\geq  \vc- \norm{\utility_{n}}\E[dist(Z_{\rounds}, CED(\game_{\alpha}))]\geq \vc-\norm{\utility_{n}}\delta.  \nonumber
\end{align}
We conclude choosing $\delta=\frac{\epsilon}{\norm{\utility_{n}}}$.
\end{proof}

Now, if the learners play No-Regret instead of No-Internal regret, using the same argument we can prove that the optimizer can guarantee at least the Hannan pessimistic Stackelberg value.

\begin{theorem} \label{HanStack}
Let $\vh$ be the Hannan pessimistic Stackelberg value of the game $\game$. If the learners are playing No-Regret, then for every $\epsilon>0$ the optimizer can guarantee at least $(\vh-\epsilon)\rounds$ expected cumulative utility for $\rounds$ large enough.
\end{theorem}

\begin{proof}
Let $(\caction,\alpha)\in\Delta(\setactions_{-n}) \times \mixedset_n$ be an Hannan Stackelberg equilibrium. Fix $\epsilon>0$. Suppose that at each round the optimizer plays the strategy $\alpha$. Since the followers play no regret; for every $i\in\setplayers$ and every $a\in\setactions_i$
\begin{equation} \label{NIR2.1}
\frac{1}{\rounds}\sum_{t=1}^{\rounds}(\utility_{\player}(a,\paction_{-i}^t)-\utility_{\player}(\paction_{i}^t\paction_{-i}^t))\leq o(1) \qquad \Prob_{\sigma} \ a.s.
\end{equation}
By Corollary \ref{CorSLLN} and inequality \eqref{NIR2.1}; repeating the passage in Theorem \ref{corrStack}, we obtain that for every $i\in\setplayers$ and for every $a\in\setactions_{\player}$
$$\frac{1}{\rounds}\sum_{t=1}^{\rounds}(\utility_{\player}(a,\paction_{-i,n}^t,\alpha)-\utility_{\player}(\paction_{i},\paction_{-i,n}^t,\alpha))\leq o(1) \qquad \Prob_{\sigma} \ a.s.$$
Hence, denoting with $Z_M \in \Delta(\setactions_{-n})$ the empirical distribution; for every $i\in\setplayers$ and for every $a\in\setactions_{\player}$
$$\sum_{\paction_{-n}\in\setactions_{-n}} Z_{\rounds}(\paction_{-n})(\utility_{\player}(a,\paction_{-i,n},\alpha)-\utility_{\player}(\paction_{-n},\alpha))\leq o(1) \qquad \Prob_{\sigma} \ a.s.$$
So, 
$$dist(Z_{\rounds}, \Han(\game_{\alpha}))\xrightarrow{M\to\infty} 0 \qquad \Prob_{\sigma} \ a.s.$$

By dominated convergence theorem 
$$\E[dist(Z_{\rounds}, \Han(\game_{\alpha}))]\xrightarrow{M\to\infty} 0.$$
Denote by $\Psi_t$ the projection of $Z_t$ to the set $\Han(\game_{\alpha})$, then repeating the inequality used in the proof of Theorem \ref{corrStack}
\begin{align}
\E\left[ \frac{1}{\rounds}\sum_{t=1}^{\rounds}\utility_n(\paction_{-n}^t, \paction_n^t)\right]&\geq \E[\utility_n(\Psi_{\rounds},\alpha)] - \norm{\utility_{n}}\E\left[dist(Z_t, \Han(\game_{\alpha}))\right]  \geq \nonumber \\
&\geq  \vh- \norm{\utility_{n}}\E[dist(Z_t, \Han(\game_{\alpha}))]\geq \vh-\norm{\utility_{n}}\delta.  \nonumber
\end{align}
We conclude choosing $\delta=\frac{\epsilon}{\norm{\utility_{n}}}$.
\end{proof}

We conclude observing that Corollary \ref{CorSLLN} allows us to adapt the proof of Theorem \ref{corrStack} and Theorem \ref{HanStack} to the payoff instead of the expected payoff.

\begin{theorem} \label{corrStack2}
Let $\vc$ ($\vh$) be the correlated (Hannan) pessimistic Stackelberg value of the game $\game$. If the learners are playing No-Internal Regret (No-Regret), then for $M$ large enough the optimizer can guarantee a.s. at least $\rounds\vc-o(\rounds)$ ($\rounds\vh-o(\rounds)$) cumulative utility.
\end{theorem}

\begin{proof}
We prove the theorem in the case of no-internal regret. The case of no-regret is analogous. Let $(\caction,\alpha)\in\Delta(\setactions_{-n}) \times \mixedset_n$ be a correlated Stackelberg equilibrium. Suppose that at each round the optimizer plays the strategy $\alpha$. Since the followers play no-internal regret; as we have shown in the Theorem \ref{corrStack}
\begin{equation} \label{distcorr}
dist(Z_{\rounds}, CED(\game_{\alpha}))\xrightarrow{M\to\infty} 0 \qquad \Prob_{\sigma} \ a.s.
\end{equation}
Now
\begin{align}
 \frac{1}{\rounds}\sum_{t=1}^{\rounds}\utility_n(\paction_{-n}^t, \paction_n^t)&= \frac{1}{\rounds}\sum_{t=1}^{\rounds}(\utility_n(\paction_{-n}^t, \paction_n^t)-\utility_n(\paction_{-n}^t, \alpha))+\frac{1}{\rounds}\sum_{t=1}^{\rounds}\utility_n(\paction_{-n}^t, \alpha)\nonumber\\
 &\geq  \vc +\frac{1}{\rounds}\sum_{t=1}^{\rounds}(\utility_n(\paction_{-n}^t, \paction_n^t)-\utility_n(\paction_{-n}^t, \alpha))- \norm{\utility_{n}}\cdot dist(Z_{\rounds}, CED(\game_{\alpha})) \nonumber \\
 &= \vc -o(1) \quad \Prob_{\sigma} \ a.s.  \nonumber
\end{align}
Where in the last equality we use Corollary \ref{CorSLLN} and \eqref{distcorr}.
\end{proof}

\section{Acknowledgements}
 I am grateful to my advisors J\'er\^{o}me Renault and Fabien Gensbittel for helpful discussions. This research has benefited from the financial support of the ANR (Programme d’Investissement
d’Avenir ANR-17-EURE-0010), and the AI Interdisciplinary Institute ANITI, which is funded by the French ”Investing for the Future - PIA3” program under the Grant
agreement ANR-19-PI3A-0004.

\nocite{*}
\bibliographystyle{unsrt}  
\bibliography{bibliografiaa}  

\end{document}